\documentclass[a4paper,fleqn]{cas-dc}

\usepackage[authoryear,longnamesfirst]{natbib}

\usepackage[english]{babel}
\usepackage[utf8]{inputenc}
\usepackage{times}
\usepackage{soul}
\usepackage{amsmath}
\usepackage{amssymb}
\usepackage{amsthm}
\usepackage{url}
\usepackage{booktabs}
\usepackage{fontawesome}
\usepackage{utfsym}
\usepackage{wasysym}
\usepackage{makecell}
\usepackage{algorithm}
\usepackage{algpseudocode}
\usepackage{pifont}
\newtheorem{example}{Example}
\newtheorem{theorem}{Theorem}

\newtheorem{lemma}{Lemma}
\newtheorem{definition}{Definition}
\newtheorem{observation}{Observation}

\newtheorem{invariant}{Invariant}
\usepackage{comment}
\usepackage[small]{caption}
\usepackage[switch]{lineno}
\usepackage{xcolor}
\usepackage{tikz}
\usetikzlibrary{positioning}
\usepackage{caption}

\def\tsc#1{\csdef{#1}{\textsc{\lowercase{#1}}\xspace}}
\tsc{WGM}
\tsc{QE}
\shorttitle{Verification of SD-EF and SD-EF1 in Time Proportional to Input Size}
\shortauthors{K.-W. Choi}

\ExplSyntaxOn
    \cs_gset:Npn \__first_footerline: { }
\ExplSyntaxOff
\begin{document}
\let\WriteBookmarks\relax
\def\floatpagepagefraction{1}
\def\textpagefraction{.001}

\title [mode = title]{Verification of Stochastic Dominance Envy-Freeness in Time Proportional to Input Size}  


\author[1]{Kui-Wang Choi}[orcid=0009-0009-8723-0223]

\ead{kuiwchoi2-c@my.cityu.edu.hk}

\affiliation[1]{organization={Department of Computer Science, City University of Hong Kong},country={Hong Kong}}

\begin{abstract}
We present a time-optimal algorithm for verifying Stochastic Dominance Envy-Freeness (SD-EF) and its relaxation, SD-EF up to one good (SD-EF1), in the fair division of indivisible goods. By leveraging a single-pass prefix-dominance check per agent and lazy-initialization, we reduce the verification complexity from the previously known $\mathcal{O}(n^2m)$ by \cite{aziz2016} to $\mathcal{O}(nm)$. Given that the input preference matrix is of size $nm$, our algorithm is asymptotically optimal with respect to the input size.
\end{abstract}

% Keywords

\maketitle

\section{Introduction}
The fair division of indivisible goods is a fundamental problem at the intersection of economics, computer science, and multi-agent systems \citep{lipton2004,budish2011}. Traditional fair division models typically assume that agents have complete cardinal valuations, exact numerical values representing their utilities for each item. In real-world applications, however, eliciting exact numerical utilities from human agents is cognitively demanding, prone to inconsistencies, and often impractical. A more realistic and widely adopted alternative assumes that agents can express only ordinal preferences, providing a simple ranking of items from most to least preferred \citep{bouveret2010}.

Evaluating fairness under ordinal preferences requires criteria that rely solely on ranking structure, rather than on hidden numerical values. A standard approach is to use Stochastic Dominance (SD), which compares the ``shapes" of bundles based on the accumulation of highly ranked items \citep{aziz2015}. Because Stochastic Dominance Envy-Freeness (SD-EF) is often impossible to guarantee when goods are indivisible (for instance, allocating a single valuable item between two agents), we must rely on its natural relaxation: Stochastic Dominance Envy-Freeness up to One Good (SD-EF1). Under SD-EF1, any envy an agent feels towards another can be eliminated by the hypothetical removal of a single item from the envied agent's bundle.

While computing fair allocations is a primary focus of the field, the verification of fairness, determining whether a given allocation satisfies a specific property, is an equally critical algorithmic task. Verification algorithms serve as essential subroutines in automated auditing tools and larger mechanism design frameworks.

In this paper, we focus on the computational efficiency of the verification process for both SD-EF and SD-EF1. While previous work by \citet{aziz2016} established that these properties can be verified in polynomial time, their approach requires $\mathcal{O}(n^2 m)$ operations, as it explicitly checks each pair of agents. In large-scale systems where the number of agents $n$ is substantial, this quadratic dependence on $n$ represents a significant overhead.

We close this gap by providing an algorithm that verifies SD-EF and SD-EF1 in $\mathcal{O}(nm)$ time. Our approach departs from the pairwise comparison paradigm. Instead, we demonstrate that by iterating once over an agent's preference list, we can simultaneously track its relationship with all other agents. 

Since the input to the problem, which is a preference matrix and an allocation, is of size $\mathcal{O}(nm)$, our algorithm is optimal with respect to the input size. This result provides a definitive complexity bound for verifying these core ordinal fairness notions.

The remainder of this paper is organized as follows. In Section 2, we include related works on stochastic dominance fairness. In Section 3, we formally define the fair division model and the stochastic dominance relations. In Section 4, we present our main algorithm, prove its correctness, and analyze its complexity.

\section{Related Works}
The study of fair division with indivisible goods under ordinal preferences has gained significant traction, as eliciting complete cardinal utilities from agents is often impractical. \cite{bouveret2010} laid crucial groundwork in this area by formalizing the concepts of possible and necessary fairness under ordinal preferences. They used SCI-nets to compute necessary envy-free allocations, establishing key existence results, particularly for the two-agent ($n=2$) case. Building on the challenges posed by ordinal preference structures, \cite{aziz2015} comprehensively investigated fair assignment algorithms, focusing on the characterization and computational complexity of achieving proportionality and envy-freeness when only ordinal rankings are available.

Given that envy-freeness is often impossible to guarantee with indivisible goods, subsequent research has shifted toward verifying relaxed fairness notions or exploring specific domain settings. \cite{aziz2016} tackled the verification problem for Necessary-EF (Consequently, SD-EF for additive valuations) specifically in the assignment problem, and provided an efficient $\mathcal{O}(n^2m)$ verification algorithm for linear ordering. Exploring stronger relaxations, \cite{hosseini2020} demonstrated that allocations satisfying both Envy-Freeness up to any good (EFX) and Pareto Optimality (PO) can be computed in polynomial time if the agents' ordinal preferences are strictly lexicographical.

\section{Preliminaries}
First, \(\forall k \in \mathbb{N}\), denote \([k] := \{1, \ldots, k\}\).

We study fair division with indivisible goods, denoted by the tuple $\mathcal{I} = \left(N, O, (\succ_1, \ldots, \succ_n), A\right)$.

We have a set of \(N = [n]\) agents, and a set of \(m\) items \(O\).

For each agent, there is an additive valuation function \(v_i: 2^{O} \longrightarrow \mathbb{R}\) over the items, where for \(S \subseteq O\), \(v_i(S) = \sum_{g \in S} v_i(\{g\})\). We slightly abuse the notation and assume that \(v_i(g) = v_i(\{g\})\).

Define \(\succcurlyeq\) (resp. $\succ$) as a weak (resp. strict) total-ordering over the goods in \(O\) induced by \(v_i\), where, for all \(g\), \(g' \in O\), \(g \succcurlyeq_i g'\) if and only if \(v_i(g) \geq v_i(g')\) (resp. $g \succ_i g'$ if and only if $v_i(g) > v_i(g')$). Each agent has a weak total-order preference for goods, i.e. \(g_1 \succcurlyeq_i \ldots \succcurlyeq_i g_m\) (resp. \(g_1 \succ_i \ldots \succ_i g_m\)). Also, we denote the ordering of goods for agent $i$ as $g_{i, 1}, \ldots, g_{i, m}$.

An allocation \(A = (A_1, \dots, A_n)\) is defined as an ordered partial partition of \(O\), where \(A_i\) is the goods allocated to agent \(i\), such that \(\forall i, j \in [N]\) where \(i \neq j\), \(A_i \cap A_j = \emptyset\) and \(\bigcup_{i \in \mathcal{N}} A_i = O\). 

We assume the input is a matrix with size \(nm\). The \((i, j)\)-th entry denotes the \(j\)-th preference good for agent \(i\).

When valuations are ordinal, we cannot compare bundles using absolute numerical values. Instead, we use Stochastic Dominance (SD), which provides a robust way to compare sets of items based on their rank distributions.
\begin{definition} [Stochastic Dominance Relation (SD)]
    Define two bundles \(X, Y \subseteq O\). \(X \succcurlyeq_i^{SD} Y\) if \(\forall g \in O\),
    \[
    |\{g': g' \succcurlyeq_i g\} \cap X| \geq |\{g': g' \succcurlyeq_i g\} \cap Y|.
    \]
\end{definition}

In other words, for every prefix of agent $i$'s preference list, $X$ contains at least as many items as $Y$. Crucially, $X \succcurlyeq_i^{SD} Y$ if and only if $v_i(X) \geq v_i(Y)$ for every additive valuation function consistent with the preference ordering $\succ_i$.Using this relation, we can define a strong notion of fairness where no agent prefers another's bundle under any consistent valuation.

\begin{definition} [Stochastic Dominance Envy-Freeness (SD-EF)] 
An allocation \(A = (A_1, \dots, A_n)\) is stochastic dominance envy-free (SD-EF) if \(\forall i, j \in N\), \(A_i \succcurlyeq_i^{SD} A_j\).
\end{definition}

While SD-EF is a powerful fairness guarantee, it is highly restrictive. To relax this notion, we consider a relaxation based on the ``up to one good" (EF1) principle.

\begin{definition} [Stochastic Dominance Envy-Freeness Up to One Good (SD-EF1)] 
An allocation \(A = (A_1, \dots, A_n)\) is stochastic dominance envy-free up to one good (SD-EF1) if \(\forall i, j \in N\), \(\exists g \in A_j\) such that \(A_i \succcurlyeq_i^{SD} A_j \setminus \{g\}\).
\end{definition}

\section{Main Result}
Before we present our algorithm, we illustrate the algorithm by \cite{aziz2016}, which verifies SD-EF for linear ordering preferences in $\mathcal{O}(n^2m)$ time.

First, they iterate through each ordered pair of agents $i$ and $j$. Then, they iterate through each good from agent $i$'s preference list, from the most preferred to the least preferred. They maintain two arrays that denote the number of goods allocated to $i$ and $j$ that have a preference not less than the current iterated good.

The allocation violates SD-EF if the number of goods allocated to agent $i$ is less than the number of goods allocated to agent $j$ at any point during the iteration of the goods in the preference list. We define this condition as the \textit{Prefix-Dominance Condition}.

It is easy to see that their approach is exactly the definition of SD-EF. We provide an example to make understanding easier.

\begin{example}
    Assume there are two agents, $m$ goods, and both agents' preferences are $g_1 \succ \ldots \succ g_m$. For demonstration purposes, assume that odd-indexed goods are allocated to agent $1$, and even-indexed goods are allocated to agent $2$.

    In the iteration of agent $1$, we iterate from $g_{1, 1}$ to $g_{1, m}$. We maintain $cnt_1$ and $cnt_2$ of size $m$ initialized with $0$.

    For $g_{1, 1}$, $cnt_1[1] = 1$ as it is allocated to agent $1$.

    For $g_{1, 2}$, $cnt_2[2] = 1$, as it is allocated to agent $2$. Since $cnt_1[2] \geq cnt_2[2]$, it does not fail the test.

    Continuing the process shows that this allocation passes the SD-EF test.

    In the iteration of agent $2$, we iterate from $g_{2, 1}$ to $g_{2, m}$. We maintain $cnt_1$ and $cnt_2$ of size $m$ initialized with $0$.

    For $g_{2, 1}$, $cnt_1[1] = 1$ as it is allocated to agent $1$. However, since $cnt_2[1] = 0 < cnt_1[1]$, it fails the test. Therefore, this instance is not SD-EF.
\end{example}

However, their algorithm's bottleneck is in the iteration over pairs of agents, resulting in $\mathcal{O}(n^2)$. We break through this bottleneck by iterating over a single agent, considering their preference list, and verifying fairness against all other agents, all at once, without additional time.

We first describe key lemmas and intuition for strict total-ordering preferences, then continue with weak total-ordering preferences.

We propose the key lemma for checking the \textit{Prefix-Dominance Condition}.
\begin{lemma} \label{lem:1}
    For any two agents $i$ and $j$, if $A_i \not\succcurlyeq^{SD}_i A_j$, i.e., $\exists g \in O$ s.t. 
    \[
    |\{g': g' \succcurlyeq_i g\} \cap A_i| < |\{g': g' \succcurlyeq_i g\} \cap A_j|,
    \]
    then $g \in A_j$.

    In other words, during the iteration through agent $i$'s preference list, if \textit{Prefix-Dominance Condition} is violated at $g$, then $g$ must not be allocated to agent $i$.
\end{lemma}
\begin{proof}
    We prove by contradiction.

    First, consider $g_{i, 1}$. If \textit{Prefix-Dominance Condition} is violated, and $g_{i, 1}$ is allocated to agent $i$, $|\{g': g' \succcurlyeq_i g_{i, 1}\} \cap A_i| = 1$. However, $\forall j \in N \setminus \{i\}$, $A_j = \emptyset$, which implies $|\{g': g' \succcurlyeq_i g\} \cap A_j| = 0$, and it does not violate the condition.

    Then, assume \textit{Prefix-Dominance Condition} is violated, and $\exists k \in [m] \setminus \{1\}$ s.t. 
    $|\{g': g' \succcurlyeq_i g_{i, k}\} \cap A_i| < |\{g': g' \succcurlyeq_i g_{i, k}\} \cap A_j|$ and $|\{g': g' \succcurlyeq_i g_{i, k-1}\} \cap A_i| \geq |\{g': g' \succcurlyeq_i g_{i, k-1}\} \cap A_j|$. However, 
    \[
    |\{g': g' \succcurlyeq_i g_{i, k}\} \cap A_i| = |\{g': g' \succcurlyeq_i g_{i, k-1}\} \cap A_i| + 1 \geq\]
    \[
    |\{g': g' \succcurlyeq_i g_{i, k}\} \cap A_j| = |\{g': g' \succcurlyeq_i g_{i, k-1}\} \cap A_j|,
    \]
    which also does not violate the condition.

    Thus, our result follows.
\end{proof}

Next, notice that the arrays used to maintain the number of goods allocated to agents $i$ and $j$ can be reduced to two integer counters, or to an integer array of length $2$.

We show how to apply Lemma~\ref{lem:1} for reducing the time complexity. First, we iterate over each agent and its preference list, from the most preferred good to the least preferred. Second, we maintain a goods counter of length $n$ that keeps track of the number of goods allocated to each agent. Lastly, during the iteration of the preference list, we use Lemma~\ref{lem:1} to check the condition between two agents only.

Then, we show how to extend this result to SD-EF1. In particular, we show how to extend our algorithm so that it does not require modifying the \textit{Prefix-Dominance Condition}. To begin with, we highlight a fact based on the property of SD-EF1.
\begin{lemma} \label{lem:2}
    Consider two agents $i$ and $j$. Let agent $j$'s bundle be non-empty, and $g$ be the most preferred good in agent $j$'s bundle for agent $i$, i.e., $\not\exists g' \in A_j$ s.t. $g' \succ_i g$.

    If $A_i \not\succcurlyeq_i^{SD} A_j \setminus \{g\}$, then $\not\exists g' \in A_j$ where $A_i \not\succcurlyeq_i^{SD} A_j \setminus \{g'\}$.

    In other words, when considering SD-EF1, we should always consider the theoretical removal of the most preferred good for agent $i$.
\end{lemma}
\begin{proof}
    It is intuitive, since if there exists any good whose removal eliminates envy, then removing the most preferred good also eliminates envy, because $A_j \setminus \{g\}$ is the worst among all single-good removals.
\end{proof}

Therefore, during the iteration over an agent's preference list, when we encounter a good that is the first (most preferred) good in another agent's bundle, we ignore it and continue.

However, the results above raise two questions. First, the goods counter of length $n$ must be initialized after each agent's preference list has been processed. Thus, the approach above takes $\mathcal{O}(n^2 + nm)$ time. \textit{How can we eliminate the overhead of initializing the counter} would be the first question. Second, in the verification of SD-EF1, \textit{how do we know that the current good is the first good we encounter for an agent?}

Surprisingly, both questions can be answered using the same approach. We use lazy initialization with a dirty bit array to resolve both problems.

Intuitively, we maintain a dirty bit array, where the $k$-th entry can tell us ``have we iterated to a good that has been allocated to agent $k$ in the current preference list iteration". If the answer is yes, we continue with the normal procedure. Otherwise, we lazily initialize the goods counter. In addition, if we are verifying SD-EF1, we also know that the current good is the first good we encounter for an agent.

The dirty bit array can be maintained efficiently by keeping track of which agent's preference list each entry was last used in.

Therefore, verifying SD-EF and SD-EF1 under strict total-ordering preferences can be done in $\mathcal{O}(nm)$. Building on this result, we extend to weak total-ordering preferences.

The key lemma, Lemma~\ref{lem:1} breaks down in weak total-ordering preferences, because we need to process all goods of the same preference before we can verify \textit{Prefix-Dominance Condition}. In the worst case, all $n$ agents have been allocated a good with the same preference, resulting in an additional $n-1$ checks. Other direct approaches, such as sorting preferences, also result in additional time. \textit{How to process goods of the same preference efficiently} is our next problem.

To resolve this, we propose the following observation.
\begin{observation} \label{lem:3}
    Assume we are now iterating through the preference list for agent $i$. Let agent $j \neq i$ be an agent that has the most preferred good allocated to them from $g_{i, 1}$ to $g_{i, k}$. If \textit{Prefix-Dominance Condition} is violated after the iteration of $g_{i, k}$, goods that are allocated to agent $j$ are more than goods that are allocated to agent $i$ between $g_{i, 1}$ to $g_{i, k}$.

    In other words, after the iteration of $g_{i, k}$, let agent $j$ be the agent with the largest goods counter entry. If \textit{Prefix-Dominance Condition} is violated, the goods counter entry for agent $j$ is more than that of agent $i$.
\end{observation}

Lemma~\ref{lem:3} implies that we only need to check between the largest goods counter entry and agent $i$ after each bundle of the same preference goods, which can be done in constant time.

Thus, we provide the high-level idea of our algorithm. First, we initialize the goods counter and the dirty bit array. Then we iterate over each agent and their preference list. We process the same preference goods in a bundle while maintaining maximum goods counter entry. After each bundle, we check \textit{Prefix-Dominance Condition}.

If \textit{Prefix-Dominance Condition} is violated during the execution of the algorithm, then the allocation is not SD-EF or SD-EF1. Otherwise, the allocation is SD-EF or SD-EF1.

We prove the algorithm's correctness and time complexity in the following theorem.

\begin{theorem}
    Verifying SD-EF and SD-EF1 with $n$ agents and $m$ items with weak total-ordering preferences can be done in $\mathcal{O}(nm)$.
\end{theorem}
\begin{proof}
    For agent $i$, we partition the items according to $i$'s preference list. Let there be $p$ partitions $(P_1, \ldots, P_p)$, such that $\forall q, r \in [p]$ where $q < r$, $\forall s \in P_{q}, \forall t \in P_{r}$, $t \succ_i s$, and $\forall u \in [p]$, $\forall s, t \in P_{u}$, $s \sim_i t$. For simplicity, assume that for any $q$, $g \succcurlyeq_i P_q$ means all goods that are preferred not less than any good in $P_q$, which also includes all goods in $P_q$.

    Define $\textit{count}$ and $\textit{dirty}$ as integer arrays of length $n$, denoting the goods counter and the dirty bit array, respectively. By the definition of Stochastic Dominance Envy-Freeness, we have the following invariant.
    \begin{invariant} \label{inv:1}
        After the iteration through the first $k$ partition of $i$'s preferences, the following holds for every agent $j$.

        \textbf{Case 1: $j$ = $i$.} $\textit{count}[i] = |\{g': g' \succcurlyeq_i P_k\} \cap A_i|$.

        \textbf{Case 2: $j \neq i$, and we are verifying SD-EF.} If $\textit{dirty}[j] \neq i$, $\textit{count}[j]$ is invalid. Otherwise, $\textit{count}[j] = |\{g': g' \succcurlyeq_i P_k\} \cap A_j|$.

        \textbf{Case 3: $j \neq i$, and we are verifying SD-EF1.} If $\textit{dirty}[j] \neq i$, $\textit{count}[j]$ is invalid. Otherwise, $\exists g'' \in A_j$, $\textit{count}[j] = |\{g': g' \succcurlyeq_i P_k\} \cap \{A_j \setminus \{g''\}\}| = |\{g': g' \succcurlyeq_i P_k\} \cap A_j| - 1$.
    \end{invariant}
    \begin{proof}
        We prove by induction. For the base case, we start by initializing agent $i$ by setting $\textit{count}[i]$ to 0 and $\textit{dirty}[i]$ to $i$. Since no goods have been allocated to agent $i$, $|A_i \cap \emptyset| = 0$, which is equal to $\textit{count}[i]$. Therefore, the base case holds.

        Assume that there exists $t$ such that the invariant holds for partitions $P_t'$, $\forall t' \leq t$. Let there be $k_1$ goods in $P_{t+1}$ that are allocated to agent $i$, and $k_2$ goods in $P_{t+1}$ that are allocated to an agent $o \neq i$.

        For agent $i$, in the algorithm, $\textit{count}[i]$ is increased by $k_1$, and since by the induction step, at $t$, $\textit{count}[i] = |\{g': g' \succcurlyeq_i P_t\} \cap A_i|$. When we include the $k_1$ items to agent $i$'s bundle, we have $|\{g': g' \succcurlyeq_i P_{t+1}\} \cap A_i| = |\{g': g' \succcurlyeq_i P_t\} \cap A_i| + k_1$, thus the induction step holds for agent $i$.

        For agent $o \neq i$, if $\textit{dirty}[j] \neq i$, we initialize to $1$ for SD-EF ($0$ for SD-EF1), then increment one for each of the remaining $k_2 - 1$ items; hence after the entire batch, $\textit{count}[j]$ is equal to $k_2$ (resp. $k_2 - 1$). Since $|\{g': g' \succcurlyeq_i P_{t+1}\} \cap A_j| = k_2$, the induction step holds for this case.
        
        Otherwise, in the algorithm, $\textit{count}[j]$ is increased by $k_2$, and since by the induction step, at $t$, $\textit{count}[j] = |\{g': g' \succcurlyeq_i P_t\} \cap A_j|$. When we include the $k_2$ items to agent $j$'s bundle, we have $|\{g': g' \succcurlyeq_i P_{t+1}\} \cap A_j| = |\{g': g' \succcurlyeq_i P_t\} \cap A_j| + k_2$ for SD-EF ($|\{g': g' \succcurlyeq_i P_{t+1}\} \cap A_j| - 1 = |\{g': g' \succcurlyeq_i P_t\} \cap A_j| + k_2 - 1$ for SD-EF1), thus the induction step holds for agent $j$.

        Consequently, our result follows.
    \end{proof}

    Next, assuming agent $j \neq i$ has the maximum $\textit{count}[j]$. Then, by Invariant~\ref{inv:1}, the violation of the \textit{Prefix-Dominance Condition} is $\forall k \neq i$, $\textit{count}[i] < \textit{count}[k] \leq \textit{count}[j]$. Thus, we only need to verify \textit{Prefix-Dominance Condition} for agent $j$.

    Therefore, our algorithm correctly verifies SD-EF and SD-EF1.

    Last, we provide complexity analysis. First, initializing $\textit{gets}$, $\textit{dirty}$, and $\textit{count}$ takes $\mathcal{O}(n)$ in total. Filling in $\textit{gets}$ takes $\mathcal{O}(m)$.

    Then, for each iteration of each agent, for each good in the preference list, initialize $\textit{dirty}$ and $\textit{count}$ for agent $i$ takes $\mathcal{O}(1)$. Since we assume that the preference for each agent is stored in a matrix, accessing the $k$-th item in an agent's preference list, conditions, accessing, and setting variables can be done in $\mathcal{O}(1)$ as well. We denote the above process as $\textit{f}$.

    This entire process takes $\mathcal{O}(nm + \sum_{i \in N} \sum_{k \in [m]} f) = \mathcal{O}(nm)$. Our result follows.
\end{proof}

\section{Discussion}
In this work, we provided an optimal verification algorithm for Stochastic Dominance Envy-Freeness and Envy-Freeness Up to One Good that runs in $\mathcal{O}(nm)$, improving from previous $\mathcal{O}(n^2m)$ by \cite{aziz2016} using lazy initialization with version arrays, effectively closing the theoretical and practical gap.

Future research can explore whether this linear-time complexity can be maintained for more complex settings, such as the division of chores or mixed manna, where the signs of item utilities are heterogeneous, and the prefix-dominance conditions become more intricate.

\section{Acknowledgement}
I would like to thank Professor Minming Li for reviewing a draft of this paper and Yachao Yan for helpful discussions.

\printcredits

\bibliographystyle{cas-model2-names}

\bibliography{bib}


\clearpage
\appendix
\section{Appendix}
\subsection{Algorithm}
\begin{algorithm}
    \caption{SD-EF and SD-EF1 Verification}
    \label{alg:sdef1_verify}
    \textbf{Input}: $\mathcal{I} = \left(N, O, (\succcurlyeq_1, \ldots, \succcurlyeq_n)\right)$\\
    \textbf{Output}: \texttt{TRUE} if $A$ is SD-EF/SD-EF1, \texttt{FALSE} otherwise
\begin{algorithmic}[1]
    \State Initialize integer array $\textit{gets}$ of size $m$ with $0$
    \For{$i = 1$ to $n$}
        \For{$g \in A_i$}
            \State $\textit{gets}[g] \gets i$
        \EndFor
    \EndFor
    \State Initialize integer array $\textit{dirty}$ of size $n$ with $0$
    \State Initialize integer array $\textit{count}$ of size $n$ with $0$
    \For{$i = 1$ to $n$}
        \State Initialize $\textit{violates}$ with \texttt{FALSE}
        \State $\textit{dirty}[i] \gets i$
        \State $\textit{count}[i] \gets 0$
        \State $ptr \gets -1$
        \State $k \gets 1$
        \For{$k = 1$ to $m$}
            \If{$k > 1$ and $g_{i, k-1} \succ g_{i, k}$} 
                \If{$\textit{count}[i] < ptr$} \Comment{\textit{Prefix-Dominance Condition}}
                    \State $\textit{violates} \gets$ \texttt{TRUE}
                \EndIf
            \EndIf
            \State $owner \gets \textit{gets}[g_{i,k}]$
            \If{$\textit{dirty}[owner] \neq i$}
                \State $\textit{dirty}[owner] \gets i$
                \If{Verifying SD-EF}
                    \State $\textit{count}[owner] \gets 1$
                \Else
                    \State $\textit{count}[owner] \gets 0$
                \EndIf
            \Else
                \State $\textit{count}[owner] \gets \textit{count}[owner] + 1$
            \EndIf
            \If{$owner \neq i$}
                \State $ptr \gets MAX(ptr, \textit{count}[owner])$
            \EndIf
        \EndFor
        \If{$ptr \neq -1$ and $\textit{count}[i] < ptr$} \Comment{\textit{Prefix-Dominance Condition} after the last item}
            \State $\textit{violates} \gets$ \texttt{TRUE}
        \EndIf
        \If{$\textit{violates}$}
            \State \Return \texttt{FALSE}
        \EndIf
    \EndFor
    \State \Return \texttt{TRUE}
\end{algorithmic}
\end{algorithm}

\subsection{Example}
We provide two examples for verifying SD-EF1 under our algorithm: one that passes the algorithm test and one that fails it.
    \begin{example}
        Assume we have three agents with identical preferences, and seven goods. Let $g_1 \succcurlyeq \ldots \succcurlyeq g_7$ for all agents.

        Let $A_1 = \{g_1, g_4, g_7\}$, $A_2 = \{g_2, g_5\}$, and $A_3 = \{g_3, g_6\}$.

        $\textit{dirty}$ is initialized as $\{i, 0, 0\}$. $\textit{count}$ is initialized as $\{0, 0, 0\}$.

        First, we iterate through each agent $i$, starting from agent $1$. $\textit{dirty} = \{1, 0, 0\}$. $\textit{count} = \{0, 0, 0\}$. $ptr = -1$.

        Next, we iterate through each item. For $g_1$, $owner = 1$. Since $\textit{dirty}[owner] = i = 1$, $\textit{count} = \{1, 0, 0\}$.

        For $g_2$, $owner = 2$. Since $\textit{dirty}[owner] \neq i$, $\textit{count} = \{1, 0, 0\}$, which is equivalent to the theoretical removal of $g_2$ with respect to agent $1$. $\textit{dirty}[owner] = 1$. $ptr = 0$.

        For $g_3$, $owner = 3$, a similar process as above is done.

        For $g_4$, $owner = 1$, thus $\textit{count} = \{2, 0, 0\}$.

        For $g_5$, $owner = 2$. Since $\textit{dirty}[owner] = i$, $\textit{count} = \{2, 1, 0\}$. $ptr = 1$.

        For the remaining goods, a similar process is done.
        
        After all iterations of the goods, we have $\textit{dirty} = \{1, 1, 1\}$, $\textit{count} = \{3, 1, 1\}$, and $ptr = 1$.

        The next agent we iterate is agent $2$. $\textit{dirty} = \{1, 2, 1\}$ as we update $\textit{dirty}[i]$ only. $\textit{count} = \{3, 0, 1\}$ as we update $\textit{count}[i]$ only. $ptr = -1$.

        For $g_1$, $owner = 1$. Since $\textit{dirty}[owner] = 1 \neq i = 2$, $\textit{count} = \{0, 0, 1\}$. $\textit{dirty} = \{2, 2, 1\}$. $ptr = 0$.

        For $g_2$, $owner = 2$. Since $\textit{dirty}[owner] = 1 \neq i = 2$, $\textit{count} = \{0, 1, 1\}$. $\textit{dirty} = \{2, 2, 2\}$.

        For $g_3$, $owner = 3$, a similar process as above is done.

        For $g_4$, $owner = 1$. Since $\textit{dirty}[owner] = i$, $\textit{count} = \{1, 1, 0\}$. $ptr = 1$.

        A similar process can be applied to the remaining goods.

        After all iterations of the goods, we have $\textit{dirty} = \{2, 2, 2\}$, $\textit{count} = \{2, 2, 1\}$, and $ptr = 2$.

        For agent $3$, a similar process can be obtained. Therefore, we can see that this example passes our algorithm test, and it is indeed SD-EF1.
    \end{example}

    \begin{example}
        Assume we have three agents with identical preferences, and seven goods. Let $g_1 \succcurlyeq \ldots \succcurlyeq g_7$ for all agents.

        Let $A_1 = \{g_1, g_2, g_7\}$, $A_2 = \{g_3, g_5\}$, and $A_3 = \{g_4, g_6\}$.

        $\textit{dirty}$ is initialized as $\{i, 0, 0\}$. $\textit{count}$ is initialized as $\{0, 0, 0\}$.

        First, we iterate through each agent $i$, starting from agent $1$. $\textit{dirty} = \{1, 0, 0\}$. $\textit{count} = \{0, 0, 0\}$. $ptr = -1$.

        Next, we iterate through each item. For $g_1$, $owner = 1$. Since $\textit{dirty}[owner] = i = 1$, $\textit{count} = \{1, 0, 0\}$.

        For $g_2$, $owner = 1$. Since $\textit{dirty}[owner] = i$, $\textit{count} = \{2, 0, 0\}$. $\textit{dirty} = \{1, 0, 0\}$.

        For $g_3$, $owner = 2$. Since $\textit{dirty}[owner] \neq i$, $\textit{count} = \{2, 0, 0\}$. $\textit{dirty} = \{1, 1, 0\}$. $ptr = 0$.

        For $g_4$, $owner = 3$, a similar process as above is obtained.

        For $g_5$, $owner = 2$. Since $\textit{dirty}[owner] = i$, $\textit{count} = \{2, 1, 0\}$. $ptr = 1$.

        For the remaining goods, a similar process is done. 
        
        After all iterations of the goods, we have $\textit{dirty} = \{1, 1, 1\}$, $\textit{count} = \{3, 1, 1\}$, and $ptr = 1$.

        The next agent we iterate is agent $2$. $\textit{dirty} = \{1, 2, 1\}$ as we update $\textit{dirty}[i]$ only. $\textit{count} = \{3, 0, 1\}$ as we update $\textit{count}[i]$ only. $ptr = -1$.

        For $g_1$, $owner = 1$. Since $\textit{dirty}[owner] = 1 \neq i = 2$, $\textit{count} = \{0, 0, 1\}$. $\textit{dirty} = \{2, 2, 1\}$. $ptr = 0$

        For $g_2$, $owner = 1$. Since $\textit{dirty}[owner] = i$, $\textit{count} = \{1, 0, 1\}$. $\textit{dirty} = \{2, 2, 1\}$. $ptr = 2$.
        
        Since $\textit{count}[1]$ is valid ($\textit{dirty}[1] = i$), and $ptr > \textit{count}[2]$, this instance fails our algorithm test, and indeed, this instance is not SD-EF1.
    \end{example}

\end{document}